\documentclass{article}

\usepackage{iclr2027_conference,times}

\usepackage[T1]{fontenc}
\usepackage[utf8]{inputenc}
\usepackage{microtype}
\usepackage{needspace}
\usepackage{amsmath,amssymb,amsfonts,amsthm,mathtools}
\usepackage{aliascnt}
\usepackage{bm}
\allowdisplaybreaks
\usepackage{graphicx}
\usepackage{array,booktabs}
\usepackage{enumitem}
\usepackage{placeins}
\usepackage{flafter} 
\usepackage{etoolbox}
\usepackage[ruled,vlined]{algorithm2e}
\SetAlFnt{\small}
\SetAlCapFnt{\small}
\SetAlCapNameFnt{\small}
\SetAlCapHSkip{0pt}
\IncMargin{-\parindent}

\usepackage{xcolor}
\usepackage{pifont}
\newcommand{\singleloopYes}{\textcolor[rgb]{0,0.8,0}{\ding{51}}}
\newcommand{\singleloopNo}{\textcolor[rgb]{1,0,0}{\ding{55}}}
\usepackage{hyperref}
\usepackage{url}
\usepackage[capitalize,noabbrev,nameinlink]{cleveref}
\usepackage{etoc}
\hypersetup{hidelinks,pdftitle={A Horizon-Independent Regret Bound for Optimistic Hedge in General-Sum Games},pdfauthor={Junsoo Ha}}
\theoremstyle{plain}
\newtheorem{theorem}{Theorem}[section]

\newaliascnt{lemma}{theorem}
\newtheorem{lemma}[lemma]{Lemma}
\aliascntresetthe{lemma}

\newaliascnt{proposition}{theorem}

\aliascntresetthe{proposition}

\newaliascnt{corollary}{theorem}

\aliascntresetthe{corollary}

\theoremstyle{definition}
\newaliascnt{definition}{theorem}
\newtheorem{definition}[definition]{Definition}
\aliascntresetthe{definition}

\newaliascnt{assumption}{theorem}

\aliascntresetthe{assumption}

\newaliascnt{example}{theorem}

\aliascntresetthe{example}

\theoremstyle{remark}
\newtheorem{remark}{Remark}[section]

\theoremstyle{plain}
\usepackage{thm-restate}
\newcommand{\printrepeatedstatement}[2]{%
  \begin{NoHyper}#2*\end{NoHyper}%
}

\crefname{theorem}{theorem}{theorems}
\Crefname{theorem}{Theorem}{Theorems}
\crefname{lemma}{lemma}{lemmas}
\Crefname{lemma}{Lemma}{Lemmas}
\crefname{proposition}{proposition}{propositions}
\Crefname{proposition}{Proposition}{Propositions}
\crefname{corollary}{corollary}{corollaries}
\Crefname{corollary}{Corollary}{Corollaries}
\crefname{remark}{remark}{remarks}
\Crefname{remark}{Remark}{Remarks}
\crefname{definition}{definition}{definitions}
\Crefname{definition}{Definition}{Definitions}
\crefname{equation}{equation}{equations}
\Crefname{equation}{Equation}{Equations}
\crefname{section}{section}{sections}
\Crefname{section}{Section}{Sections}

\newcommand{\R}{\mathbb{R}}
\newcommand{\E}{\mathbb{E}}
\newcommand{\1}{\mathbf{1}}
\newcommand{\Reg}{\operatorname{Reg}}
\newcommand{\Var}{\operatorname{Var}}
\newcommand{\osc}{\operatorname{osc}}
\newcommand{\diag}{\operatorname{Diag}}
\newcommand{\ip}[2]{\left\langle #1,#2\right\rangle}
\newcommand{\norm}[1]{\left\lVert #1\right\rVert}

\newcommand{\dd}{\mathrm{d}}
\newcommand{\appref}[1]{\hyperref[#1]{Appendix~\ref*{#1}}}
\newcolumntype{L}[1]{>{\raggedright\arraybackslash}p{#1}}

\title{A Horizon-Independent Regret Bound for\\Optimistic Hedge in General-Sum Games}

\author{Junsoo Ha \\
\texttt{junsoo.ha.contact@gmail.com}}

\iclrfinalcopy
\makeatletter
\patchcmd{\@maketitle}
  {\lhead{Published as a conference paper at ICLR 2027}}
  {\lhead{}}
  {}{\PackageError{manuscript}{Could not remove publication-status header}{Check the title macro in the supplied style.}}
\makeatother

\begin{document}
\etocdepthtag.toc{main}

\maketitle

\begin{abstract}
\looseness=-1
Can simple learning rules keep their regret bounded in self-play? Recent work achieves constant regret bounds through modified regularization and higher-order prediction \citep{LiuFarinaOzdaglar2026,AbbadiLarakiMertikopoulos2026}. Yet for Optimistic Hedge, arguably the most canonical method in games, the best known individual regret bound remains logarithmic \citep{Ha2026LogHedge}. In this work, we prove that plain Optimistic Hedge with a constant step size can attain $O_{n,d}(1)$ individual regret in general-sum games with $n$ players and $d=(d_1,\ldots,d_n)$ actions, under expected loss-vector feedback. As a corollary, its time-averaged play enjoys an $O_{n,d}(1/T)$ coarse correlated equilibrium (CCE) gap. Our analysis represents Optimistic Hedge as a real-analytic recurrence on a compact space, which yields an exact finite-order difference relation that eliminates horizon dependence. Our proof hinges on nonconstructive Noetherianity argument of \cite{Frisch1967}, so the $(n,d)$-dependence remains implicit.
\end{abstract}

\section{Introduction}
\label{sec:introduction}

A central question in game theory \citep{VonNeumannMorgenstern1944,Nash1950} is whether independent learning can lead to an equilibrium when opponents are adapting. Such interdependent learning dynamics underpin modern applications such as poker self-play \citep{BrownSandholm2019}, adversarial training \citep{GoodfellowEtAl2014}, and game-theoretic approaches to language-model alignment \citep{MunosEtAl2024,RossetEtAl2024,WuEtAl2024SPPO,GolzEtAl2025,WuEtAl2026MNPO}.

No-regret learning \citep{CBL2006} offers a decentralized approach: players seek low regret against the best fixed action in hindsight \citep{Hannan1957}. Indeed, sublinear external regret drives time-averaged play toward an equilibrium \citep{MoulinVial1978,HartMasColell2001}.

While classical online learning guarantees $O(\sqrt T)$ regret against arbitrary losses \citep{FreundSchapire1997}, self-play adds an additional structure: a common update rule shared by all players. Optimistic methods exploit this structure and use recent losses to predict the next one and improve regret bounds \citep{RakhlinSridharan2013,SALS2015}. Optimistic Hedge, the most canonical optimistic method for games, has recently been shown to enjoy polylogarithmic \citep{DFG2021} and logarithmic \citep{Ha2026LogHedge} regret bounds. On the other hand, recent alternatives have shown that modified regularization and higher-order prediction methods can produce horizon-independent regret bounds \citep{LiuFarinaOzdaglar2026,AbbadiLarakiMertikopoulos2026}. In this paper, we ask the following natural question:
\begin{center}
\itshape
Can Optimistic Hedge attain horizon-independent individual regret in general-sum games?
\end{center}

\paragraph{Contributions.}
Our main result gives an affirmative answer. Our contributions are two-fold.
\begin{enumerate}[label=(\roman*),leftmargin=1.7em,itemsep=0pt,topsep=0pt,parsep=0pt]
\item \looseness=-1 \textbf{An analytic finite-order difference relation.} Based on the high-order smoothness framework \citep{DFG2021} and the weighted-gap reduction \citep{Ha2026LogHedge}, we represent Optimistic Hedge as an analytic recurrence on a compact domain. A Noetherianity argument then yields an exact finite-order difference relation controlling the first-order weighted-gap differences.
\item \textbf{A horizon-independent regret bound.} Based on the finite-order difference relation and a simple finite-difference interpolation argument, we obtain $O_{n,d}(1)$ individual regret and an $O_{n,d}(1/T)$ CCE gap. The dimension dependence remains nonconstructive.
\end{enumerate}

Section~\ref{sec:related} reviews prior work; Sections~\ref{sec:setup}--\ref{sec:main-results} introduce the setting and main result. Section~\ref{sec:analysis} explains the proof strategy, with detailed proofs in the appendices.

\section{Related work}
\label{sec:related}

\textbf{No-regret learning and equilibrium.}
The link between no-regret learning and equilibrium computation dates back to \cite{Blackwell1956} and \cite{Hannan1957}. External regret compares cumulative loss with the best fixed action in hindsight. When all players achieve sublinear regret, averaging yields approximate Nash equilibria in two-player zero-sum games \citep{FreundSchapire1999} and CCE in general-sum games \citep{MoulinVial1978,HartMasColell2001}. Controlling recommendation-dependent deviations via internal or swap regret yields correlated equilibrium \citep{HartMasColell2000,BlumMansour2007}.

\textbf{Individual regret in self-play.}
Improving upon standard no-regret learning algorithms, such as follow-the-regularized-leader (FTRL), optimistic algorithms exploit predictable loss sequences \citep{ChiangEtAl2012,RakhlinSridharan2013}. Notably, the seminal work of \cite{SALS2015} showed that optimistic follow-the-regularized-leader (OFTRL) and optimistic mirror descent (OMD) can exploit the self-play-induced predictability in games to achieve $O(\sqrt n\log d_{\max}T^{1/4})$ individual regret, where \mbox{$d_{\max}:=\max_i d_i$} is the maximum number of actions across players. For Optimistic Hedge, \cite{ChenPeng2020} obtained $O(\log^{5/6}d_{\max}T^{1/6})$ in two-player games, and \cite{DFG2021} established $O(n\log d_i\log^4T)$ in $n$-player games based on a high-order smoothness analysis.

A recent line of work improves regret bounds by modifying regularization, step sizes, or prediction methods. Log-regularized lifted OFTRL (LRL-OFTRL) \citep{FarinaEtAl2022} enlarges the regularized strategy space and achieves $O(nd_{\max}\log T)$ in finite general-sum games, and extends to general convex strategy sets. Cautious Optimistic MWU (COMWU) \citep{SoleymaniEtAl2025,SoleymaniEtAl2025Fast} achieves $O(n\log^2 d_{\max}\log T)$ through adaptive step sizes. ECHO-OFTRL \citep{LiuFarinaOzdaglar2026} and HOOD \citep{AbbadiLarakiMertikopoulos2026} combine higher-order prediction with lifted regularization to obtain $O(n^{21}\log^4 d_{\max})$ and $O(n^3\log^2 d_{\max})$ regret bounds, respectively. \Cref{tab:positioning} compares these individual regret bounds.

\textbf{Other regret notions and update rules.}
A parallel line studies swap regret to obtain guarantees for correlated equilibrium (CE) \citep{Aumann1974}. \cite{AnagnostidesEtAl2022CE} extend the high-order smoothness framework of \cite{DFG2021} to obtain polylogarithmic internal and swap regret. Log-barrier regularization \citep{AnagnostidesEtAl2022LogSwap} and hybrid regularization \citep{Tsuchiya2026Hybrid} yield logarithmic and sublogarithmic swap regret, respectively. Clairvoyant Multiplicative Weights Update (CMWU) \citep{PiliourasEtAl2022} obtains horizon-independent regret via implicit updates; however, its uncoupled implementation guarantees this bound only on a selected subsequence. For two-player general-sum games, \cite{TaoZheng2026} achieve $O(\log d_{\max})$ alternating regret and $O(\log d_{\max}/T)$ CCE error using sequential updates, unlike our simultaneous external-regret setting.

\textbf{Closest prior work.}
The closest works are \cite{DFG2021} and \cite{Ha2026LogHedge}. We build on the high-order smoothness framework of \cite{DFG2021} and the recurrence-based weighted-gap analysis of \cite{Ha2026LogHedge}. In particular, \cite{Ha2026LogHedge} represents Optimistic Hedge as a recurrent map acting on a state space equipped with a fixed Euclidean norm. Crucially, the high-order estimates of \cite{Ha2026LogHedge} leave a remainder that grows with the horizon at a fixed step size. In this work, we instead obtain an exact relation among difference orders and eliminate the horizon dependence in our final regret bound.

\begin{table}[!t]
\centering
\setlength{\tabcolsep}{3pt}
\renewcommand{\arraystretch}{1.05}
\caption{Individual external regret in full-information general-sum normal-form games; $d_{\max}:=\max_i d_i$. The Chen--Peng bound assumes two players; LRL-OFTRL also applies to convex games. ``Single loop'' means closed-form updates without inner optimization or root finding, excluding feedback cost. The OFTRL/OMD check refers to its entropy-regularized Optimistic Hedge instantiation. The notation $O_{n,d}(1)$ hides dependence on the numbers of players and actions, but not on $T$.}
\label{tab:positioning}
\footnotesize
\begin{tabular}{@{}L{0.25\linewidth}L{0.245\linewidth}>{\centering\arraybackslash}p{0.115\linewidth}L{0.335\linewidth}@{}}
\toprule
Reference & Algorithm & Single loop & Regret bound \\
\midrule
\citet{SALS2015} & OFTRL/OMD & \singleloopYes & $O(\sqrt n\log d_{\max}T^{1/4})$ \\
\citet{ChenPeng2020} & Optimistic Hedge & \singleloopYes & $O(\log^{5/6}d_{\max}T^{1/6})$ \\
\cite{DFG2021} & Optimistic Hedge & \singleloopYes & $O(n\log d_i\log^4T)$ \\
\citet{FarinaEtAl2022} & LRL-OFTRL & \singleloopNo & $O(nd_{\max}\log T)$ \\
\citet{SoleymaniEtAl2025,SoleymaniEtAl2025Fast} & COMWU & \singleloopNo & $O(n\log^2 d_{\max}\log T)$ \\
\citet{LiuFarinaOzdaglar2026} & ECHO-OFTRL & \singleloopNo & $O(n^{21}\log^4 d_{\max})$ \\
\citet{AbbadiLarakiMertikopoulos2026} & HOOD & \singleloopNo & $O(n^3\log^2 d_{\max})$ \\
\cite{Ha2026LogHedge} & Optimistic Hedge & \singleloopYes & $O(\sqrt n\log d_i\log T)$ \\
This work & Optimistic Hedge & \singleloopYes & $O_{n,d}(1)$ \\
\bottomrule
\end{tabular}
\end{table}

\section{Preliminaries}
\label{sec:setup}

We study full-information repeated play in finite general-sum games within the classical no-regret learning framework \citep{Hannan1957,CBL2006}.

\paragraph{Notation.}
For an integer \(q\ge1\), write \([q]:=\{1,\ldots,q\}\) and define a probability simplex \(\Delta_q:=\{x\in\R_+^q:\sum_ax_a=1\}\). Fix \(n\ge1\) players with action counts \(d_i\ge1\), action sets \(A_i=[d_i]\), and mixed strategies \(x_i\in\Delta_{d_i}\). Set action dimensions \(d:=(d_1,\ldots,d_n)\) and \(d_{\max}:=\max_i d_i\). The notation \(O_{n,d}(\cdot)\) hides dependence on \(n,d\), but not on the horizon; all logarithms are natural.

\paragraph{Loss functions.}
We denote a loss table by \(\ell:\prod_{j=1}^nA_j\to[0,1]^n\) and define player \(i\)'s expected loss \(\ell_i:\prod_{j=1}^n\Delta_{d_j}\to[0,1]\) as \(\ell_i(x):=\E_{a\sim\bigotimes_jx_j}[\ell(a)]_i\). For a strategy profile \(x=(x_1,\ldots,x_n)\), write \(x_{-i}:=(x_j)_{j\ne i}\). We define the expected loss vector \(\ell_i(x_{-i})\in[0,1]^{d_i}\) by
\begin{equation}
[\ell_i(x_{-i})]_{a_i}
:=\sum\nolimits_{a_{-i}\in\prod_{j\ne i}A_j}
[\ell(a_i,a_{-i})]_i\prod\nolimits_{j\ne i}x_{j,a_j}.
\label{eq:loss-vector}
\end{equation}
Hence, the identity \(\ell_i(x)=\ip{x_i}{\ell_i(x_{-i})}\) holds. We also write \(\ell(x):=(\ell_i(x_{-i}))_{i=1}^n\) for the stacked expected-loss profile; the argument distinguishes it from the loss table. At round \(t\), write \(\ell^t:=\ell(x^t)\); player \(i\) incurs \(\ip{x_i^t}{\ell_i^t}\) and observes the full vector \(\ell_i^t=\ell_i(x_{-i}^t)\).

Computing a Nash equilibrium is PPAD-complete even for two-player general-sum games \citep{DaskalakisGoldbergPapadimitriou2009,ChenDengTeng2009}. Coarse correlated equilibrium (CCE) \citep{MoulinVial1978} is a tractable relaxation: every Nash equilibrium induces a CCE, but a CCE allows correlated actions and tests only unconditional deviations. A CCE can be computed by linear programming with a time complexity polynomial in \(n\prod_i d_i\) and the payoff bit length \citep{PapadimitriouRoughgarden2008}.

\begin{definition}[\(\varepsilon\)-approximate coarse correlated equilibrium]
\label{def:cce}
For a joint distribution \(\mu\) on \(\prod_jA_j\), define its CCE gap by
\begin{equation}
\operatorname{Gap}_{\rm CCE}(\mu)
:=\max_{i\in[n]}\max_{a_i'\in A_i}
\E_{a\sim\mu}\!\left[[\ell(a)]_i-[\ell(a_i',a_{-i})]_i\right].
\label{eq:cce-gap}
\end{equation}
For \(\varepsilon\ge0\), \(\mu\) is an \(\varepsilon\)-approximate CCE if \(\operatorname{Gap}_{\rm CCE}(\mu)\le\varepsilon\).
\end{definition}

No-regret learning approaches CCE by controlling each player's external regret \citep{HartMasColell2001,CBL2006}, which is defined as the difference between its cumulative loss and that of its best fixed action in hindsight.

\begin{definition}[External regret]
\label{def:external-regret}
At horizon \(T\ge1\), player \(i\)'s individual external regret is
\begin{equation}
\Reg_i(T):=\sum_{t=1}^T\ip{x_i^t}{\ell_i^t}-\min_{a_i'\in A_i}\sum_{t=1}^T\ell_{i,a_i'}^t.
\label{eq:regret}
\end{equation}
\end{definition}

Let \(\widehat\mu_T:=T^{-1}\sum_{t=1}^T\bigotimes_{j=1}^n x_j^t\) be the time-averaged joint distribution, the uniform mixture of the product distributions played at each round. Under $\widehat\mu_T$, each fixed-action deviation gain is its average gain over the played rounds. As a result, the following standard result holds:
\begin{equation}
\operatorname{Gap}_{\rm CCE}(\widehat\mu_T)=\frac1T\max_{i\in[n]}\Reg_i(T).
\label{eq:cce-generic}
\end{equation}
Therefore, individual-regret bounds directly control the CCE approximation error.

\paragraph{Optimistic Hedge.}
Hedge is an entropy-regularized follow-the-regularized-leader (FTRL) algorithm \citep{FreundSchapire1997}. Optimistic Hedge uses \(\ell_i^t\) to predict \(\ell_i^{t+1}\) \citep{RakhlinSridharan2013}. All players use a constant step size \(\eta>0\), with \(\ell_i^0=0\) and \(x_i^1=\1_{d_i}/d_i\), where \(\1_{d_i}\) is the all-ones vector. Eliminating cumulative losses from its closed-form update gives
\begin{equation}
x_{i,a}^{t+1}=\frac{x_{i,a}^t\exp\{-\eta(2\ell_{i,a}^t-\ell_{i,a}^{t-1})\}}
{\sum_{b\in A_i}x_{i,b}^t\exp\{-\eta(2\ell_{i,b}^t-\ell_{i,b}^{t-1})\}},\qquad t\ge1,\quad a\in A_i.
\label{eq:omwu}
\end{equation}
For the update in \Cref{eq:omwu}, \cite{ChenPeng2020} obtain $O(\log^{5/6}d_{\max}T^{1/6})$ individual regret in two-player games. For $n$ players, \cite{DFG2021} establish $O(n\log d_i\log^4T)$, and \cite{Ha2026LogHedge} sharpens this to $O(\sqrt n\log d_i\log T)$. We show that plain Optimistic Hedge can attain horizon-independent $O_{n,d}(1)$ regret for a sufficiently small constant step size.

\section{Main Result}
\label{sec:main-results}

The following theorem shows that Optimistic Hedge admits a horizon-independent individual-regret bound at sufficiently small constant step sizes.

\begin{restatable}[Horizon-independent individual regret]{theorem}{horizonRegretRestated}
\label{thm:constant}
Fix \(n\) players with action counts \(d_1,\ldots,d_n\). There exist \(\eta_\star>0\) and \(C_\star<\infty\), depending only on \((n,d_1,\ldots,d_n)\), such that, for every game in the setup of \Cref{sec:setup} and every common constant step size \(0<\eta\le\eta_\star\), the Optimistic Hedge iterates satisfy, for every \(i\in[n]\) and \(T\ge1\),
\begin{equation}
\Reg_i(T)\le\frac{\log d_i}{\eta}+C_\star\eta.
\label{eq:constant-regret}
\end{equation}
\end{restatable}

The logarithmic guarantee for Optimistic Hedge uses a horizon-tuned constant step size \citep{Ha2026LogHedge}. Specifically, the parameterized bound of \cite{Ha2026LogHedge} leaves a remainder proportional to $\eta T\exp\{-c/(\sqrt n\,\eta)\}$ (for a universal $c>0$), which is exponentially small in the inverse step size but still grows with the horizon. In contrast, \Cref{thm:constant} yields a horizon-independent $O_{n,d}(1)$ individual-regret bound.

The following corollary translates this regret bound into a convergence guarantee for time-averaged play, using \Cref{eq:cce-generic}.

\begin{restatable}[Coarse-correlated-equilibrium convergence]{corollary}{cceRestated}
\label{cor:cce}
Under the conditions of \Cref{thm:constant}, the time-averaged joint distribution \(\widehat\mu_T\) satisfies, for every \(T\ge1\),
\begin{equation}
\operatorname{Gap}_{\rm CCE}(\widehat\mu_T)
\le \frac1T\max_{i\in[n]}\left\{\frac{\log d_i}{\eta}+C_\star\eta\right\}.
\label{eq:cce-rate}
\end{equation}
\end{restatable}

\Cref{cor:cce} improves the horizon dependence of the $O(\sqrt n\log d_{\max}\log T/T)$ CCE bound of \cite{Ha2026LogHedge} from $\log T/T$ to $1/T$. The improvement is in $T$, not in game dimensions: our $O_{n,d}(1/T)$ guarantee leaves $n,d$ unquantified. We leave explicit bounds on $\eta_\star,C_\star$ and a dimension-explicit horizon-independent individual regret bound to future work.

\par\addvspace{0.5\baselineskip}
\begin{remark}[Comparison with modified dynamics]
\label{rem:modified-dynamics}
ECHO-OFTRL achieves $O(n^{21}\log^4 d_{\max})$ individual external regret \citep{LiuFarinaOzdaglar2026}, and HOOD achieves $O(n^3\log^2 d_{\max})$ \citep{AbbadiLarakiMertikopoulos2026}. Hybrid-regularized OFTRL with the Blum--Mansour reduction gives $O(nd_{\max}^2\sqrt{\log d_{\max}\log T})$ swap regret and a correlated-equilibrium guarantee \citep{Tsuchiya2026Hybrid}. However, these modified dynamics involve inner optimization or root-finding subroutines. Our result instead establishes a horizon-independent regret bound for plain Optimistic Hedge, which admits a closed-form update.
\end{remark}

\section{Proof Overview}
\label{sec:analysis}

We follow the high-order smoothness framework of \cite{DFG2021} and the weighted-gap reduction of \cite{Ha2026LogHedge}. Our real-analytic state representation of Optimistic Hedge yields an exact finite-order relation through the Noetherianity theorem of \cite{Frisch1967}. Finite-difference interpolation then converts this relation into a horizon-independent regret bound.

\subsection{Reduction to a First-Order Difference Estimate}

For $x\in\Delta_q$ and $v\in\R^q$, write the variance of a vector as $\Var_x(v):=\sum_ax_a(v_a-\ip{x}{v})^2$. \cite{Ha2026LogHedge} gives the following small-step-size specialization of Lemma~4.1 in \cite{DFG2021}.

\begin{restatable}[Optimistic regret bound {\citep[Lemma~5.1]{Ha2026LogHedge}}]{lemma}{regretSimpleRestated}
\label{lem:regret-simple}
There is a universal \(\eta_{\rm rvu}>0\) such that, for every player, horizon, and \(0<\eta\le\eta_{\rm rvu}\),
\begin{equation}
\Reg_i(T)
\le\frac{\log d_i}{\eta}
 +\frac{2\eta}{3}\sum_{t=1}^T\Var_{x_i^t}(\ell_i^t-\ell_i^{t-1})
 -\frac{\eta}{3}\sum_{t=1}^T\Var_{x_i^t}(\ell_i^{t-1}).
\label{eq:regret-simple}
\end{equation}
\end{restatable}

The first sum charges prediction error; the second supplies the negative variance term. Bounding the first by half the second, up to a horizon-independent remainder, controls regret. The time-varying quadratic form $\Var_{x_i^t}(\cdot)$ prevents a direct comparison in a fixed sequence norm; hence, we incorporate its probability weights into the vectors themselves. The pairwise identity
\begin{equation}
\Var_x(v)=\sum_{a<b}x_ax_b(v_a-v_b)^2
\label{eq:pairwise-variance}
\end{equation}
motivates the weighted pairwise loss-gap vector $h_i^t\in\R^{\binom{d_i}{2}}$ of \cite{Ha2026LogHedge}, defined for $a<b$ by
\begin{equation}
[h_i^t]_{ab}
:=\sqrt{x_{i,a}^t x_{i,b}^t}
  \bigl(\ell_{i,a}^{t-1}-\ell_{i,b}^{t-1}\bigr).
\label{eq:h}
\end{equation}
The following lemma from \cite{Ha2026LogHedge} controls the sum of prediction-error variances $\Var_{x_i^t}(\ell_i^t-\ell_i^{t-1})$ in terms of the weighted-gap differences $\Delta h_i^t:=h_i^{t+1}-h_i^t$, accounting for time-varying weights.

\begin{restatable}[Weighted-gap reduction {\citep[Lemma~5.2]{Ha2026LogHedge}}]{lemma}{weightedGapRestated}
\label{lem:weighted-gap}
There are universal constants \(C>0\) and \(\eta_{\rm gap}>0\) such that, whenever \(0<\eta\le\eta_{\rm gap}\),
\begin{equation}
\begin{aligned}
\sum_{t=1}^T\Var_{x_i^t}(\ell_i^{t-1})
 &=\sum_{t=1}^T\norm{h_i^t}_2^2,\\
\sum_{t=1}^T\Var_{x_i^t}(\ell_i^t-\ell_i^{t-1})
 &\le\frac32\sum_{t=1}^T\norm{h_i^{t+1}-h_i^t}_2^2
   +C\eta^2\sum_{t=1}^T\norm{h_i^t}_2^2+C\eta^2.
\end{aligned}
\label{eq:weighted-gap}
\end{equation}
\end{restatable}

\Cref{lem:regret-simple,lem:weighted-gap} reduce the remaining challenge to bounding the sum of first-order differences of weighted gaps $\sum_{t=1}^T\norm{\Delta h_i^t}_2^2$ by the sum of weighted gaps themselves $\sum_{t=1}^{T+1}\norm{h_i^t}_2^2$, up to a horizon-independent remainder. The following key lemma supplies this bound.

\begin{lemma}[First-order difference estimate]
\label{lem:constant-temporal}
Fix \(n\) and the action counts \(d\). There are constants \(C,C'<\infty\) and \(\eta_{\rm temp}>0\), depending only on \(n,d\), such that, for every game of these dimensions, every player \(i\in[n]\), every horizon \(T\ge1\), and every common constant step size \(0<\eta\le\eta_{\rm temp}\),
\begin{equation}
\sum_{t=1}^T\norm{h_i^{t+1}-h_i^t}_2^2
\le C\eta^2\sum_{t=1}^{T+1}\norm{h_i^t}_2^2+C'.
\label{eq:constant-temporal}
\end{equation}
\end{lemma}

Unlike the first-order difference estimate of \cite{Ha2026LogHedge}, this bound is horizon-independent: $C'$ depends only on the game dimensions $(n,d)$, but not on the horizon $T$. \Cref{lem:regret-simple,lem:weighted-gap,lem:constant-temporal} yield \Cref{thm:constant}, so our goal now reduces to proving \Cref{lem:constant-temporal}.

\subsection{Reduction to a High-Order Difference}
\label{sec:high-order-reduction}

For a sequence $y=(y^t)_{t\in\mathbb Z}$, define finite differences by $\Delta^0y^t:=y^t$, $\Delta y^t:=y^{t+1}-y^t$, and $\Delta^{k+1}y^t:=\Delta(\Delta^ky^t)$. For a sequence with a finite support, write $\norm y_{\ell_2}^2:=\sum_t\norm{y^t}^2$.

To prove \Cref{lem:constant-temporal}, we use the following lemma to control the first-order difference $\norm{\Delta y}_{\ell_2}$ in terms of the sequence norm $\norm y_{\ell_2}$ and the order-$m$ difference $\norm{\Delta^m y}_{\ell_2}$.

\begin{restatable}[Finite-difference interpolation]{lemma}{fourierToolsRestated}
\label{lem:fourier-tools}
For every finitely supported sequence \(y=(y^t)_{t\in\mathbb Z}\) in a finite-dimensional Hilbert space, the following holds for all integers \(m\ge1\) and \(0\le j\le m\):
\begin{equation}
\norm{\Delta^j y}_{\ell_2}
\le \norm{y}_{\ell_2}^{1-j/m}
 \norm{\Delta^m y}_{\ell_2}^{j/m}.
\label{eq:difference-interpolation}
\end{equation}
\end{restatable}

Intuitively, Parseval's identity expresses the difference norms in a frequency domain, and H\"older's inequality interpolates between orders $0$ and $m$. We defer the proof to \appref{app:common}.

\Cref{lem:fourier-tools} suggests the following route. Suppose the weighted gaps satisfy the finite-order relation
\begin{equation}
\Delta^m h_i^{t+m}
=\sum_{k=0}^{m-1}c_{k,t}\eta^{m-k}\Delta^k h_i^{t+k},
\qquad |c_{k,t}|\le A,\quad t\ge1,
\label{eq:target-finite-relation}
\end{equation}
where $m,A$ depend only on $(n,d)$. At each horizon, we can construct a finitely supported cutoff $y$ of $(h_i^t)$ whose error in \eqref{eq:target-finite-relation} is bounded by $O_{n,d}(\eta^{m-1/2})$ in $\ell_2$ norm (\Cref{lem:exact-boundary-extension}). Taking norms gives
\[
\norm{\Delta^m y}_{\ell_2}
\le \tfrac12\norm{\Delta^m y}_{\ell_2}
       +O_{n,d}(\eta^m)\norm y_{\ell_2}+O_{n,d}(\eta^{m-1/2}).
\]
Here we use \Cref{lem:fourier-tools} and Young's inequality. Absorbing the first term on the right and applying \Cref{lem:fourier-tools} with $j=1$ and Young's inequality once more yields
\[
\norm{\Delta y}_{\ell_2}^2
\le\norm y_{\ell_2}^{2-2/m}
    \norm{\Delta^m y}_{\ell_2}^{2/m}
\le O_{n,d}(\eta^2)\norm y_{\ell_2}^2+O_{n,d}(\eta).
\]
The cutoff construction $y$ in \Cref{lem:exact-boundary-extension} transfers this estimate back to $h_i^t$, yielding \Cref{lem:constant-temporal}, and hence, \Cref{thm:constant}. Notice that the above derivation crucially hinges on the exact finite-order relation \eqref{eq:target-finite-relation}; we prove that precise finite-order relation in the following subsections.

\subsection{State Recurrence and Finite-Order Difference Relation}
\label{sec:state-recurrence}

To obtain the finite-order relation \eqref{eq:target-finite-relation}, we write Optimistic Hedge as a recurrence on a compact state space. Specifically, we encode the scaled differences $\eta^{-k}\Delta^k h_i^{t+k}$ by analytic maps. Then, Frisch's Noetherianity theorem \citep{Frisch1967}, which guarantees finite generation of analytic spans near compact semianalytic sets, gives the finite relation we need.

\paragraph{Square-root state.}
The factor $\sqrt{x_{i,a}x_{i,b}}$ in \Cref{eq:h} is not analytic in the probabilities at the boundary. Hence, we set $r_i(a)=\sqrt{x_i(a)}$, which turns the pair weights into the polynomial products $r_i(a)r_i(b)$ and places each $r_i$ on a compact nonnegative unit sphere. We store the loss profile as $\ell=(\ell_i)_{i=1}^n$ and use the state space
\[
\mathcal S
:=\left\{(r,\ell):
 r_i\in\R_+^{d_i},\ \norm{r_i}_2=1,
 \ell_i\in[0,1]^{d_i}\ \text{for every }i
\right\}.
\]

\paragraph{State recurrence.}
For a nonzero $r_i$, define the mixed strategy $x_i(r_i)\in\Delta_{d_i}$ by $[x_i(r_i)]_a:=r_i(a)^2/\norm{r_i}_2^2$. Recall that $\ell$ denotes the stacked expected-loss map from \Cref{sec:setup}, and write $\ell(r):=\ell((x_i(r_i))_{i=1}^n)$. For $s=(r,\ell)\in\mathcal S$, define the strategy update $F_i:\mathcal S\times\R\to\R^{d_i}$ and the full state update $G:\mathcal S\times\R\to\mathcal S$ by
\begin{align}
[F_i(r,\ell;\eta)]_a
&:=\frac{\norm{r_i}_2\,r_i(a)
 \exp\{-\eta(2[\ell(r)]_i(a)-[\ell_i]_a)/2\}}
 {\left(\sum_b r_i(b)^2
 \exp\{-\eta(2[\ell(r)]_i(b)-[\ell_i]_b)\}\right)^{1/2}},
\label{eq:analytic-square-root-update}\\
G(r,\ell;\eta)&:=\left(\bigl(F_i(r,\ell;\eta)\bigr)_{i=1}^n,\ell(r)\right).
\label{eq:state-map}
\end{align}
The factor $\norm{r_i}_2$ leaves the update unchanged on $\mathcal S$ and extends it to a neighborhood of the unit spheres. We keep $\eta$ fixed under composition: $G^0(s;\eta)=s$ and $G^{j+1}(s;\eta)=G(G^j(s;\eta);\eta)$. For Optimistic Hedge, setting $r_i^t(a)=\sqrt{x_i^t(a)}$ and $s^t=(r^t,\ell^{t-1})$ gives $s^{t+1}=G(s^t;\eta)$.

\paragraph{Weighted gaps.}
We map the square-root state back to the weighted gaps with $H_{i,0}:\mathcal S\times\R\to\R^{\binom{d_i}{2}}$:
\begin{equation}
[H_{i,0}(r,\ell;\eta)]_{ab}
:=r_i(a)r_i(b)\bigl([\ell_i]_a-[\ell_i]_b\bigr),
\qquad a<b.
\label{eq:exact-observable-main}
\end{equation}
For $\eta\ne0$ and $k\ge0$, define the scaled-difference maps recursively by
\begin{equation}
H_{i,k+1}(s;\eta)
:=\frac{H_{i,k}(G^2(s;\eta);\eta)
      -H_{i,k}(G(s;\eta);\eta)}{\eta}.
\label{eq:scaled-difference-functions}
\end{equation}
Induction along $s^{t+1}=G(s^t;\eta)$ gives $\Delta^k h_i^{t+k}=\eta^k H_{i,k}(s^t;\eta)$ for $k\ge0$ and $t\ge1$.

Notice that a relation $H_{i,m}=\sum_{k<m}c_kH_{i,k}$ immediately yields the finite-order relation \eqref{eq:target-finite-relation} after multiplication by $\eta^m$. A natural question is therefore when such a relation would exist among the maps $H_{i,k}$. The following consequence of the Noetherianity theorem of \cite{Frisch1967} shows that real analyticity of these maps near a compact semianalytic set suffices; see also \cite{Lonsted1973}. Recall that a semianalytic set is locally specified by finitely many real-analytic equalities and inequalities.

\begin{restatable}[Finite relation among analytic maps]{theorem}{finiteDependenceRestated}
\label{thm:finite-dependence}
Let $K\subset\R^q$ be compact and semianalytic, and let $V$ be a finite-dimensional real vector space. For every sequence $\Phi_0,\Phi_1,\ldots$ of $V$-valued maps, each real analytic on a neighborhood of $K$, there exist an integer $m\ge1$, an open neighborhood $U$ of $K$, and scalar real-analytic functions $c_0,\ldots,c_{m-1}$ on $U$ such that the following holds on $U$:
\begin{equation}
\Phi_m=\sum_{k=0}^{m-1}c_k\Phi_k.
\label{eq:finite-dependence}
\end{equation}
\end{restatable}

Locally, analytic maps behave like polynomials: after finitely many $\Phi_k$, new maps can be expressed using earlier ones with analytic coefficients. The number of earlier maps needed locally may vary across $K$. Compactness and semianalyticity rule out unbounded growth of this number; we defer the proof to \appref{app:finite-generation}. We verify these hypotheses for $H_{i,k}$ and transfer the relation to $h_i^t$.

\subsection{Proof of the Finite-Order Difference Relation}
\label{sec:finite-relation-proof}

Recall that our goal is to derive the finite-order relation \eqref{eq:target-finite-relation} for the weighted gaps $h_i^t$ in arbitrary games with dimensions $(n,d)$. To make the finite relation produced by \Cref{thm:finite-dependence} hold uniformly over such games, we include the entire loss table in the arguments of our maps. Let
\[
\mathcal G:=[0,1]^{n\prod_{j=1}^n d_j},
\qquad
K_\star:=\mathcal S\times\mathcal G\times\{0\},
\]
where $\mathcal G$ is the set of loss tables. We choose the singleton step size space $\{0\}$ because we only need the relation for sufficiently small step sizes. With $g\in\mathcal G$ fixed, we suppress it throughout this subsection:
\[
F_i(r,\ell;\eta)\equiv F_i(r,\ell;g,\eta),
\qquad
G(s;\eta)\equiv G(s;g,\eta),
\qquad
H_{i,k}(s;\eta)\equiv H_{i,k}(s;g,\eta).
\]
We now verify the hypotheses of \Cref{thm:finite-dependence} for $\Phi_k=H_{i,k}$ and prove the finite-order relation \eqref{eq:target-finite-relation}.

\paragraph{Step 1: analyticity of $G$ and $H_{i,k}$.}
We first study the analyticity of $G$ and its behavior at $\eta=0$.

\begin{restatable}[Analyticity of the state update]{lemma}{analyticStateRepresentationRestated}
\label{lem:analytic-state-representation}
The state update $G(s;\eta)$ is jointly real analytic in the state, the loss-table entries, and $\eta$ on one common neighborhood of \(K_\star\). The state space $\mathcal S$ is invariant. Moreover, $G(r,\ell;0)=(r,\ell(r))$ and $G^2(s;0)=G(s;0)$ throughout this neighborhood.
\end{restatable}

At $\eta=0$, $G$ only replaces the stored losses by $\ell(r)$, so applying it twice changes nothing further. Thus the zero-step identity makes the numerator in \Cref{eq:scaled-difference-functions} vanish at $\eta=0$. By \Cref{lem:analytic-state-representation}, the numerator is analytic whenever $H_{i,k}$ is, so the quotient extends analytically through $\eta=0$. Thus the recursion preserves analyticity from the polynomial base $H_{i,0}$, giving \Cref{lem:scaled-difference-analytic}.

\begin{restatable}[Analyticity of scaled finite differences]{lemma}{scaledDifferenceAnalyticRestated}
\label{lem:scaled-difference-analytic}
For every fixed $k\ge0$, the function $H_{i,k}$ in \Cref{eq:scaled-difference-functions} extends real analytically through $\eta=0$, jointly in the state and loss-table entries, on a neighborhood of \(K_\star\). The neighborhood may depend on $k$, but not on the state or game.
\end{restatable}

\paragraph{Step 2: finite generation and uniform coefficient bounds.}
Now that we have established the analyticity of $H_{i,k}$, we can apply \Cref{thm:finite-dependence} to $H_{i,0},H_{i,1},\ldots$. For some finite $m$, \Cref{thm:finite-dependence} gives analytic coefficient functions $c_0,\ldots,c_{m-1}$ on a common neighborhood of $K_\star$ such that
\begin{equation}
H_{i,m}=\sum_{k=0}^{m-1}c_kH_{i,k}.
\label{eq:finite-state-formula-main}
\end{equation}
Evaluating \eqref{eq:finite-state-formula-main} along the Optimistic Hedge trajectory gives the finite-order relation we have sought:
\[
\Delta^m h_i^{t+m}
=\eta^m H_{i,m}(s^t;\eta)
=\sum_{k=0}^{m-1}c_k(s^t;\eta)\eta^{m-k}\Delta^k h_i^{t+k}.
\]
Compactness of $K_\star$ puts $\mathcal S\times\mathcal G\times[0,\eta_i]$ inside the common neighborhood for some $\eta_i>0$ depending on $(n,d)$. Continuity bounds $c_k$ and $H_{i,k}$ on this compact set; invariance of $\mathcal S$ makes these bounds valid at every round and produces the following lemma.

\begin{restatable}[Exact finite-order relation]{lemma}{finiteGapRelationRestated}
\label{lem:finite-gap-relation}
Fix the number of players $n$, number of actions $d=(d_1,\ldots,d_n)$, and a player $i$. There are an integer $m\ge1$, constants $A,M<\infty$, and $\eta_i>0$, uniform over all games of these dimensions, such that, for every such game, $0<\eta\le\eta_i$, and $t\ge1$, there exist scalars $c_{0,t},\ldots,c_{m-1,t}$, possibly depending on the game, $\eta$, and $t$, satisfying
\begin{align}
\Delta^m h_i^{t+m}
&=\sum_{k=0}^{m-1}c_{k,t}\eta^{m-k}\Delta^kh_i^{t+k},
&|c_{k,t}|&\le A,
\label{eq:exact-relation-main}\\*
\norm{\Delta^kh_i^{t+k}}_2
&\le M\eta^k,
&0\le k&\le m.
\label{eq:scaled-difference-size-main}
\end{align}
\end{restatable}

\Cref{lem:finite-gap-relation} gives the precise finite-order relation \eqref{eq:target-finite-relation} that we needed in \Cref{sec:high-order-reduction}. In particular, it provides a uniform bound on the coefficients $c_{k,t}$ and does not include any error term that accumulates over the horizon $T$. The exact order $m$ and constants $A,M$ remain implicit due to the nonconstructive nature of \Cref{thm:finite-dependence}. We defer the full proof to \appref{app:finite-generation}.

\subsection{From the Finite-Order Relation to the Regret Bound}
\label{sec:finite-horizon}

With the exact finite-order relation from \Cref{lem:finite-gap-relation} in hand, one final step unlocks \Cref{lem:fourier-tools}: a finitely supported cutoff of $h_i^t$ with boundary error independent of $T$. The next lemma supplies this cutoff construction to complete the argument in \Cref{sec:high-order-reduction}; we defer the proof to \appref{app:technical}.

\begin{restatable}[Finite-support cutoff]{lemma}{finiteSupportCutoffRestated}
\label{lem:exact-boundary-extension}
Fix an integer $m\ge1$ and $A,M<\infty$. For $0<\eta<1/m$, suppose a sequence $(u^t)_{t\ge1}$ in a finite-dimensional Euclidean space satisfies, for every $t\ge1$,
\begin{align}
\Delta^m u^{t+m}
&=\sum_{k=0}^{m-1}c_{k,t}\eta^{m-k}\Delta^k u^{t+k},
\qquad |c_{k,t}|\le A\quad(0\le k<m),
\label{eq:exact-boundary-recurrence-main}\\
\norm{\Delta^k u^{t+k}}
&\le M\eta^k\qquad(0\le k\le m).
\label{eq:exact-boundary-assumption-main}
\end{align}
There is a constant $C<\infty$, depending only on $m,A,M$, such that every horizon $T\ge1$ admits a finitely supported sequence $y=(y^t)_{t\in\mathbb Z}$ with
\begin{align}
\norm{\Delta y}_{\ell_2}^2
&\ge\sum\nolimits_{t=1}^{T}\norm{\Delta u^t}^2-C,
\label{eq:exact-boundary-return-main}\\
\norm y_{\ell_2}^2
&\le\sum\nolimits_{t=1}^{T+1}\norm{u^t}^2,
\label{eq:cutoff-sequence-norm}\\
\norm{\Delta^m y}_{\ell_2}
&\le A\sum\nolimits_{k=0}^{m-1}\eta^{m-k}\norm{\Delta^k y}_{\ell_2}
 +C\eta^{m-1/2}.
\label{eq:exact-norm-recurrence-main}
\end{align}
\end{restatable}

This finite-support construction is similar to that of \citet{Ha2026LogHedge}, but is tailored to preserve the finite-order relation up to a horizon-independent error. The inequalities \eqref{eq:exact-boundary-return-main}--\eqref{eq:cutoff-sequence-norm} let us transfer a first-order difference bound for $y$ back to $h_i^t$ only at a constant cost.

\begin{proof}[Proof of \Cref{lem:constant-temporal}]
We combine the finite-order relation in \Cref{lem:finite-gap-relation}, the cutoff bound \eqref{eq:exact-norm-recurrence-main} from \Cref{lem:exact-boundary-extension}, and the finite-difference interpolation in \Cref{lem:fourier-tools} to bound $\norm{\Delta y}_{\ell_2}$. We then use \eqref{eq:exact-boundary-return-main} and \eqref{eq:cutoff-sequence-norm} from \Cref{lem:exact-boundary-extension} to transfer the bound to $h_i^t$ with error independent of $T$.

Fix $i,T$ and take $m,A,M,\eta_i$ from \Cref{lem:finite-gap-relation}. For $0<\eta\le\bar\eta_i:=\frac12\min\{\eta_i,1/m,1\}$, apply \Cref{lem:exact-boundary-extension} to $(h_i^t)_{t\ge1}$ to obtain the cutoff $y$, and write $C_1\ge0$ for its constant $C$.

The cutoff bound \eqref{eq:exact-norm-recurrence-main} in \Cref{lem:exact-boundary-extension} and the finite-difference interpolation \eqref{eq:difference-interpolation} in \Cref{lem:fourier-tools} give:
\begin{align}
\norm{\Delta^m y}_{\ell_2}
&\overset{\mathclap{\text{\eqref{eq:exact-norm-recurrence-main}}}}{\le}
 A\sum_{k=0}^{m-1}\eta^{m-k}\norm{\Delta^k y}_{\ell_2}
 +C_1\eta^{m-1/2}\notag\\*
&\overset{\mathclap{\text{\eqref{eq:difference-interpolation}}}}{\le}
 A\eta^m\norm y_{\ell_2}
 +A\sum_{k=1}^{m-1}\bigl(\eta^m\norm y_{\ell_2}\bigr)^{1-k/m}
   \norm{\Delta^m y}_{\ell_2}^{k/m}
 +C_1\eta^{m-1/2}.
\label{eq:cutoff-interpolation}
\end{align}

For each term in the sum, Young's inequality gives
\begin{equation}
A\bigl(\eta^m\norm y_{\ell_2}\bigr)^{1-k/m}\norm{\Delta^m y}_{\ell_2}^{k/m}
\le \frac{1}{2m}\norm{\Delta^m y}_{\ell_2}
 +\frac{m-k}{m}(2k)^{\frac{k}{m-k}}A^{\frac{m}{m-k}}\eta^m\norm y_{\ell_2}.
\label{eq:young-absorption}
\end{equation}
Substituting \eqref{eq:young-absorption} into \eqref{eq:cutoff-interpolation} gives a coefficient $(m-1)/(2m)\le1/2$ on $\norm{\Delta^m y}_{\ell_2}$. Moving this term to the left yields the following for some constant $C_2\ge0$ depending only on $m,A,M$:
\begin{equation}
\norm{\Delta^m y}_{\ell_2}
\le C_2\eta^m\norm y_{\ell_2}+C_2\eta^{m-1/2}.
\label{eq:highest-difference-bound}
\end{equation}

\emph{Case $m=1$.} \Cref{eq:highest-difference-bound} already bounds $\norm{\Delta y}_{\ell_2}$. Squaring and using $(a+b)^2\le2a^2+2b^2$ gives
\[
\norm{\Delta y}_{\ell_2}^2
\le 2C_2^2\eta^2\norm y_{\ell_2}^2+2C_2^2\eta.
\]

\emph{Case $m\ge2$.} Invoking \Cref{lem:fourier-tools} with $j=1$ and combining it with \eqref{eq:highest-difference-bound} gives:
\begin{align}
\norm{\Delta y}_{\ell_2}^2
&\overset{\mathclap{\text{\eqref{eq:difference-interpolation}}}}{\le}
 \norm y_{\ell_2}^{2-2/m}\norm{\Delta^m y}_{\ell_2}^{2/m}
 \overset{\mathclap{\text{\eqref{eq:highest-difference-bound}}}}{\le}
 C_2^{2/m}\norm y_{\ell_2}^{2-2/m}
 \bigl(\eta^m\norm y_{\ell_2}+\eta^{m-1/2}\bigr)^{2/m}\notag\\*
&\le
 C_2^{2/m}\eta^2\norm y_{\ell_2}^2
 +C_2^{2/m}\bigl(\eta^2\norm y_{\ell_2}^2\bigr)^{1-1/m}\eta^{1/m}\notag\\*
&\le C_2^{2/m}\bigl((2-1/m)\eta^2\norm y_{\ell_2}^2+\eta/m\bigr)\notag\\*
&\le 2C_2^{2/m}\bigl(\eta^2\norm y_{\ell_2}^2+\eta\bigr).
\label{eq:first-difference-closure}
\end{align}
The third inequality uses $(a+b)^{2/m}\le a^{2/m}+b^{2/m},\ \forall m\ge2$. The fourth applies Young's inequality $(\eta^2\norm y_{\ell_2}^2)^{1-1/m}\eta^{1/m}\le(1-1/m)\eta^2\norm y_{\ell_2}^2+\eta/m$.

We now return to $h_i^t$ using \eqref{eq:exact-boundary-return-main} and \eqref{eq:cutoff-sequence-norm}. Either case and $\eta\le1$ give
\[
\sum_{t=1}^T\norm{h_i^{t+1}-h_i^t}_2^2
\le\norm{\Delta y}_{\ell_2}^2+C_1
\le 2C_2^{2/m}\eta^2\sum_{t=1}^{T+1}\norm{h_i^t}_2^2+C_1+2C_2^{2/m}.
\]
The coefficients depend only on $(n,d)$ and $i$, through $m,A,M$. Taking common upper bounds over $i$ and $\eta_{\rm temp}:=\min_i\bar\eta_i$ proves \Cref{lem:constant-temporal}.
\end{proof}

\paragraph{Completing the regret bound.}
Combining \Cref{lem:regret-simple}, \Cref{lem:weighted-gap}, and \Cref{lem:constant-temporal} directly proves \Cref{thm:constant}. We give a detailed calculation in \appref{app:main-proof}.

\section{Discussion}
\label{sec:discussion}

\paragraph{Implications.}
Notably, a recent line of work has obtained horizon-independent regret bounds in finite general-sum games by combining modified regularization with higher-order prediction \citep{LiuFarinaOzdaglar2026,AbbadiLarakiMertikopoulos2026}. However, \Cref{thm:constant} shows that plain Optimistic Hedge can already attain this horizon independence with a dimension-dependent constant step size. Our result suggests that adaptive learning rates \citep{SoleymaniEtAl2025,SoleymaniEtAl2025Fast}, hybrid regularization \citep{Tsuchiya2026Hybrid}, and higher-order predictions \citep{LiuFarinaOzdaglar2026,AbbadiLarakiMertikopoulos2026} are unnecessary for this guarantee.

\paragraph{Comparison with \cite{Ha2026LogHedge}.}
Both our analysis and the analysis of \cite{Ha2026LogHedge} build on the high-order smoothness framework of \cite{DFG2021}. Ha's centered-logit derivative bounds control each difference order separately. Choosing $m\asymp(\sqrt n\eta)^{-1}$ leaves a regret contribution proportional to $\eta T\exp\{-c/(\sqrt n\eta)\}$, leading to a logarithmic regret bound when the step size is tuned to the horizon. In our analysis, Frisch's theorem on the compact state space yields an exact relation among difference orders (\Cref{lem:finite-gap-relation}). Its order is fixed by $(n,d)$, and it has no additive error that accumulates over $T$. Finite-difference interpolation (\Cref{lem:fourier-tools}) then yields a first-order difference bound that lets us absorb prediction error into stability, up to a horizon-independent error. Intuitively, our approach thus trades Ha's explicit dimension dependence for horizon independence.

\section{Conclusion}
In this work, we proved that plain Optimistic Hedge with a constant step size can achieve horizon-independent individual regret in finite general-sum games. An exact finite-order relation for weighted gaps removes the horizon-accumulating remainder in the logarithmic analysis of \cite{Ha2026LogHedge}, and yields an $O_{n,d}(1/T)$ CCE guarantee. Quantifying the explicit dimension dependence of $\eta_\star$ and $C_\star$ in our regret bound remains open, and we leave it as an intriguing future work.

\clearpage
\subsection*{AI use statement}
The author formulated the initial horizon-independent regret hypothesis and conducted a numerical experiment to test the hypothesis empirically. The hypothesis was motivated by the close connection between optimistic and implicit update rules \citep{MokhtariOzdaglarPattathil2020}. Specifically, the authors hypothesized that a similar connection might exist between Optimistic Hedge and Clairvoyant Multiplicative Weights Update (CMWU) \citep{PiliourasEtAl2022}, and therefore, Optimistic Hedge might also enjoy horizon-independent regret bounds. Under the authors' direction, generative AI tools were used to develop and write proof arguments based on the high-order smoothness framework of \cite{DFG2021} and the weighted-gap reduction of \cite{Ha2026LogHedge}. AI tools also assisted with literature search, exposition, and LaTeX editing. The author takes responsibility for the final content, including all mathematical claims, proofs, and references.

\subsection*{Reproducibility statement}
The assumptions and algorithm are specified in \Cref{sec:setup,sec:main-results}. Complete proofs appear in \Cref{sec:analysis} and Appendices~\ref{app:main-proof}--\ref{app:technical}.

\bibliography{references}

\clearpage
\appendix
\etocdepthtag.toc{appendix}

\begingroup
\renewcommand{\contentsname}{Appendix Contents}
\etocsettagdepth{main}{-3}
\etocsettagdepth{appendix}{subsection}
\etocsetnexttocdepth{subsection}
\tableofcontents
\endgroup

\section{Proof of the main result}
\label{app:main-proof}

We derive the horizon-independent regret guarantee from the three estimates in Section~5.1.

\printrepeatedstatement{thm-constant}{\horizonRegretRestated}

\begin{proof}[Proof of \Cref{thm:constant}]
Let $\eta_{\rm rvu},\eta_{\rm gap},\eta_{\rm temp}$ be the thresholds in \Cref{lem:regret-simple,lem:weighted-gap,lem:constant-temporal}, with $\eta_{\rm temp}\le1$. Since $\norm{h_i^{T+1}}_2^2\le1/4$, the last two lemmas give
\begin{equation}
\sum_{t=1}^T\Var_{x_i^t}(\ell_i^t-\ell_i^{t-1})
\le C_0\eta^2\sum_{t=1}^T\Var_{x_i^t}(\ell_i^{t-1})+C_1,
\label{eq:constant-variance}
\end{equation}
with constants $C_0,C_1\ge0$ depending only on $n,d$. For $0<\eta\le\min\{\eta_{\rm rvu},\eta_{\rm gap},\eta_{\rm temp}\}$, substitution into \Cref{eq:regret-simple} yields
\[
\Reg_i(T)
\le\frac{\log d_i}{\eta}
 +\frac{2C_1\eta}{3}
 -\frac{\eta}{3}(1-2C_0\eta^2)\sum_{t=1}^T\Var_{x_i^t}(\ell_i^{t-1}).
\]
Set
\[
\eta_\star:=\min\!\left\{
\eta_{\rm rvu},\eta_{\rm gap},\eta_{\rm temp},
\frac{1}{\sqrt{2\max\{C_0,1\}}}
\right\},
\qquad
C_\star:=\frac{2C_1}{3}.
\]
The variance coefficient is nonpositive for every $0<\eta\le\eta_\star$, proving \Cref{eq:constant-regret}. Both constants depend only on $n,d$.
\end{proof}

\printrepeatedstatement{cor-cce}{\cceRestated}

\begin{proof}
Apply \Cref{eq:cce-generic} to \Cref{thm:constant}: for every $T\ge1$,
\[
\operatorname{Gap}_{\rm CCE}(\widehat\mu_T)
=\frac1T\max_{i\in[n]}\Reg_i(T)
\le\frac1T\max_{i\in[n]}\left\{\frac{\log d_i}{\eta}+C_\star\eta\right\}.
\]
This proves \Cref{eq:cce-rate}.
\end{proof}

\section{Preliminary lemmas}
\label{app:common}

We prove the optimistic regret bound, the weighted-gap reduction, and finite-difference interpolation. For the first two, we only use the update rule of Optimistic Hedge and boundedness of loss vectors; we use the analytic structure of self-play in \appref{app:exact}.

\subsection{Optimistic regret bound}

Fix a player $i$, write $d:=d_i$, and fix $\eta>0$. Define the soft-min potential $\Phi:\R^d\to\R$ by
\[
\Phi(z):=-\frac1\eta\log\sum_{a=1}^d e^{-\eta z_a}
\]
Its gradient gives exponential weights, while its value approximates the smallest coordinate. Since $-\Phi$ is convex, we can define its corresponding Bregman divergence for $y,z\in\R^d$ by
\[
D_{-\Phi}(y,z):=\Phi(z)+\ip{\nabla\Phi(z)}{y-z}-\Phi(y).
\]

For $v\in\R^d$, write $\osc(v):=\max_av_a-\min_av_a$ for its coordinate span. The following auxiliary bound from \cite{Ha2026LogHedge} compares each Bregman divergence with the variance under its reference distribution.

\begin{lemma}[Bregman divergence and local variance {\citep[Lemma~B.1]{Ha2026LogHedge}}]
\label{lem:softmin-variance}
Let $z,v\in\R^d$ and $R\ge0$, set $p:=\nabla\Phi(z)$, and suppose that $\osc(v)\le R$. Then
\begin{equation}
\frac\eta2 e^{-\eta R}\Var_p(v)
\le D_{-\Phi}(z+v,z)
\le\frac\eta2 e^{\eta R}\Var_p(v).
\label{eq:softmin-variance}
\end{equation}
\end{lemma}

\begin{proof}
We compare the potential's curvature along the segment from $z$ to $z+v$. Set $p_s:=\nabla\Phi(z+sv)$ for $s\in[0,1]$; these are the exponential weights along that segment. Since
\[
-\nabla^2\Phi(z+sv)
=\eta\bigl(\diag(p_s)-p_sp_s^\top\bigr),
\]
the integral remainder formula gives
\begin{equation}
D_{-\Phi}(z+v,z)
=\eta\int_0^1(1-s)\Var_{p_s}(v)\,\dd s.
\label{eq:hessian-integral}
\end{equation}
Moreover,
\[
\frac{p_s(a)}{p(a)}
=\frac{e^{-\eta s v_a}}{\sum_b p(b)e^{-\eta s v_b}},
\]
so $e^{-\eta sR}p(a)\le p_s(a)\le e^{\eta sR}p(a)$ coordinatewise. Using
\[
\Var_p(v)=\min_{c\in\R}\sum_a p(a)(v_a-c)^2,
\]
we obtain
\[
e^{-\eta R}\Var_p(v)\le\Var_{p_s}(v)\le e^{\eta R}\Var_p(v).
\]
Substitute into \Cref{eq:hessian-integral} and use $\int_0^1(1-s)\dd s=1/2$.
\end{proof}

The optimistic regret bound of \cite{Ha2026LogHedge} specializes Lemma~4.1 of \cite{DFG2021}, whose Lemma~A.5 refines Lemma~3 of \cite{RakhlinSridharan2013Online}. We reproduce its proof here for completeness.

\printrepeatedstatement{lem-regret-simple}{\regretSimpleRestated}

\begin{proof}
Suppress the player index from $x_i^t$, $\ell_i^t$, and $\Reg_i$. We first express the strategy in terms of past losses. Set $L^0:=0$, and for $t\ge1$ let
\[
L^t:=\sum_{\tau=1}^t\ell^\tau,
\qquad
q^t:=L^{t-1}+\ell^{t-1}.
\]
Here $L^t$ is the cumulative loss, and $q^t$ augments the observed cumulative loss with the latest loss vector as a prediction for round $t$. Unrolling Optimistic Hedge gives
\[
x_a^t=\frac{e^{-\eta q_a^t}}{\sum_{b\in[d]} e^{-\eta q_b^t}},
\qquad a\in[d].
\]
This discrepancy is the curvature correction in a potential increment. Using $\nabla\Phi(q^t)=x^t$ and $L^t-L^{t-1}=\ell^t$, we obtain the one-round identity and its cumulative regret consequence:
\begin{align}
D_{-\Phi}(L^t,q^t)-D_{-\Phi}(L^{t-1},q^t)
&=\ip{x^t}{\ell^t}-\Phi(L^t)+\Phi(L^{t-1}),\notag\\*
\Reg(T)
&\le\frac{\log d}{\eta}
 +\sum_{t=1}^T\left[D_{-\Phi}(L^t,q^t)-D_{-\Phi}(L^{t-1},q^t)\right].
\label{eq:regret-bregman}
\end{align}
The potential terms telescope, and $\Phi(L^T)\le\min_aL_a^T$ compares the final potential with the best fixed action. Uniform initialization gives $\Phi(0)=-(\log d)/\eta$, hence the initial cost in the second line.

Apply \Cref{lem:softmin-variance} to the positive and negative divergences in \Cref{eq:regret-bregman}. Since $L^t-q^t=\ell^t-\ell^{t-1}$ has coordinate span at most $2$, while $L^{t-1}-q^t=-\ell^{t-1}$ has span at most $1$, \Cref{eq:regret-bregman,eq:softmin-variance} give the exact inequality
\begin{equation}
\Reg(T)
\le\frac{\log d}{\eta}
 +\frac{\eta e^{2\eta}}2\sum_{t=1}^T\Var_{x^t}(\ell^t-\ell^{t-1})
 -\frac{\eta e^{-\eta}}2\sum_{t=1}^T\Var_{x^t}(\ell^{t-1}).
\label{eq:regret-exact}
\end{equation}
Choose a universal $\eta_{\rm rvu}>0$ such that $e^{2\eta}/2\le2/3$ and $e^{-\eta}/2\ge1/3$ for $0<\eta\le\eta_{\rm rvu}$. This proves \Cref{lem:regret-simple}.
\end{proof}

\subsection{Weighted pairwise loss gaps}

For completeness, we reproduce the weighted-gap reduction from Lemma~5.2 of \cite{Ha2026LogHedge}. The pairwise identity transfers the variances to a fixed Euclidean norm; the correction comes from the first-order differences of the probability weights.

\printrepeatedstatement{lem-weighted-gap}{\weightedGapRestated}

\begin{proof}
The pairwise variance identity \eqref{eq:pairwise-variance} follows by expanding $\frac12\sum_{a,b}x_ax_b(v_a-v_b)^2$ and noting that the diagonal terms vanish. Applying it to \Cref{eq:h} gives the identity in \Cref{eq:weighted-gap}.

For prediction error, the remaining obstacle is that the strategy weights change between rounds. Fix player $i$ and pair $a<b$, and write
\[
w_{ab}^t:=\sqrt{x_{i,a}^t x_{i,b}^t}.
\]
Uniform initialization and the multiplicative update keep $w_{ab}^t>0$ at every finite round, so these ratios are well defined. Their bounds remain uniform as probabilities approach zero. The optimistic loss vector $m_i^t:=2\ell_i^t-\ell_i^{t-1}$ has $\osc(m_i^t)\le3$, and \Cref{eq:omwu} gives
\begin{equation}
e^{-3\eta}\le\frac{w_{ab}^{t+1}}{w_{ab}^t}\le e^{3\eta},
\qquad
\left|\frac{w_{ab}^t}{w_{ab}^{t+1}}-1\right|\le C\eta
\label{eq:weight-ratio}
\end{equation}
for a universal $C$ and all sufficiently small $\eta$. Thus the relative change of each pair weight is $O(\eta)$.

Since $[h_i^t]_{ab}=w_{ab}^t(\ell_{i,a}^{t-1}-\ell_{i,b}^{t-1})$,
\begin{equation}
w_{ab}^t\bigl[(\ell_{i,a}^{t}-\ell_{i,b}^{t})-(\ell_{i,a}^{t-1}-\ell_{i,b}^{t-1})\bigr]
=[h_i^{t+1}-h_i^t]_{ab}
 +\left(\frac{w_{ab}^t}{w_{ab}^{t+1}}-1\right)[h_i^{t+1}]_{ab}.
\label{eq:moving-weight}
\end{equation}
The first term is the weighted-gap difference $[\Delta h_i^t]_{ab}$; the second accounts for the weight ratio $w_{ab}^t/w_{ab}^{t+1}$. Its coefficient is $O(\eta)$, so squaring yields the $O(\eta^2)$ correction in the lemma. Apply $(a+b)^2\le\frac32a^2+3b^2$ to \Cref{eq:moving-weight}, then sum over $t$ and $a<b$. Using \Cref{eq:pairwise-variance,eq:weight-ratio},
\begin{align*}
\sum_{t=1}^T\Var_{x_i^t}(\ell_i^t-\ell_i^{t-1})
&=\sum_{t=1}^T\sum_{a<b}
 \left(w_{ab}^t\bigl[(\ell_{i,a}^{t}-\ell_{i,b}^{t})
 -(\ell_{i,a}^{t-1}-\ell_{i,b}^{t-1})\bigr]\right)^2\\
&\le\frac32\sum_{t=1}^T\norm{h_i^{t+1}-h_i^t}_2^2
 +C\eta^2\sum_{t=1}^T\norm{h_i^{t+1}}_2^2\\
&\le\frac32\sum_{t=1}^T\norm{h_i^{t+1}-h_i^t}_2^2
 +C\eta^2\sum_{t=1}^T\norm{h_i^t}_2^2+C\eta^2.
\end{align*}
The last inequality uses $h_i^1=0$ and
\[
\sum_{t=1}^T\norm{h_i^{t+1}}_2^2
=\sum_{t=1}^T\norm{h_i^t}_2^2+\norm{h_i^{T+1}}_2^2
\le\sum_{t=1}^T\norm{h_i^t}_2^2+\frac14.
\]
The terminal bound follows from \(\norm{h_i^{T+1}}_2^2=\Var_{x_i^{T+1}}(\ell_i^T)\le1/4\). Thus moving the summation window contributes only the final $C\eta^2$ term, and \Cref{lem:weighted-gap} follows.
\end{proof}

\subsection{Interpolation of finite differences}

This interpolation step is independent of Optimistic Hedge. It applies because the reduction to weighted gaps has put all difference orders in one Hilbert space with a fixed norm.

\printrepeatedstatement{lem-fourier-tools}{\fourierToolsRestated}

\begin{proof}
Fourier transformation turns finite differences into multiplication, allowing all orders to be compared in one representation. Let \(\widehat y(\omega):=\sum_{t\in\mathbb Z}y^te^{-it\omega}\) be the Fourier transform and \(\lambda(\omega):=|e^{i\omega}-1|^2\) the squared magnitude of the first-order difference multiplier. Parseval's identity gives
\[
\norm{\Delta^j y}_{\ell_2}^2
=\frac1{2\pi}\int_{-\pi}^{\pi}\lambda(\omega)^j
 \norm{\widehat y(\omega)}^2\,\dd\omega.
\]
Thus each order weights the same Fourier components by a different power of $\lambda$. The cases \(j=0,m\) are immediate. For \(0<j<m\), apply H\"older's inequality with exponents \(m/(m-j)\) and \(m/j\) to
\[
\lambda^j\norm{\widehat y}^2
=\bigl(\norm{\widehat y}^2\bigr)^{1-j/m}
 \bigl(\lambda^m\norm{\widehat y}^2\bigr)^{j/m}.
\]
This proves \Cref{eq:difference-interpolation}.
\end{proof}

\section{Exact finite-order difference relation}
\label{app:exact}

We first extend the state update and scaled finite-difference maps analytically through zero step size. \appref{app:frisch-noetherian} introduces the analytic prerequisites and proves Frisch's Noetherianity theorem from them. \appref{app:finite-generation} then proves \Cref{thm:finite-dependence} and applies it uniformly over states and games to obtain \Cref{lem:finite-gap-relation}. We defer the finite-support construction to \appref{app:technical}.

\subsection{Analytic state update}

To preserve the quantifiers in \Cref{thm:constant}, we include the loss table in the compact parameter set from the outset. Define
\[
\mathcal R_i:=\{r_i\in\R_+^{d_i}:\norm{r_i}_2=1\},
\qquad
\mathcal R:=\prod_{i=1}^n\mathcal R_i,
\qquad
\mathcal V:=\prod_{i=1}^n[0,1]^{d_i}.
\]
Thus $\mathcal S=\mathcal R\times\mathcal V$, where $\mathcal V$ contains the stored loss profiles. Let $\mathcal G:=[0,1]^{n\prod_jd_j}$ be the set of loss tables, and write $g=(g_1,\ldots,g_n)\in\mathcal G$. For $r_i\ne0$, set $[x_i(r_i)]_a:=r_i(a)^2/\norm{r_i}_2^2$. For $r\in\mathcal R$, define the current expected-loss profile by
\[
[\ell_g(r)]_i(a_i)
:=\sum_{a_{-i}\in\prod_{j\ne i}A_j}
 g_i(a_i,a_{-i})\prod_{j\ne i}[x_j(r_j)]_{a_j}.
\]
Set $K_0:=\mathcal R\times\mathcal V\times\mathcal G$ and $K_\star:=K_0\times\{0\}$. We restore the loss-table argument in the maps as $F_i(r,\ell;g,\eta)$, $G(s;g,\eta)$, and $H_{i,k}(s;g,\eta)$; under $G^j(s;g,\eta)$, we keep both $g$ and $\eta$ fixed. The update is
\[
\begin{aligned}
[F_i(r,\ell;g,\eta)]_a
&:=\frac{\norm{r_i}_2\,r_i(a)
 e^{-\eta(2[\ell_g(r)]_i(a)-[\ell_i]_a)/2}}
 {\left(\sum_b r_i(b)^2
 e^{-\eta(2[\ell_g(r)]_i(b)-[\ell_i]_b)}\right)^{1/2}},\\
G(r,\ell;g,\eta)&:=\left(\bigl(F_i(r,\ell;g,\eta)\bigr)_{i=1}^n,\ell_g(r)\right).
\end{aligned}
\]

\printrepeatedstatement{lem-analytic-state-representation}{\analyticStateRepresentationRestated}

\begin{proof}
We verify joint analyticity, invariance, and the zero-step identity in turn. On a sufficiently small neighborhood of $\mathcal R$, every $\norm{r_i}_2$ stays bounded away from zero. Hence $x_i(r_i)$ and $\ell_g(r)$ are rational, and therefore jointly real analytic, in $(r,g)$. At $\eta=0$, the expression under the denominator square root in \Cref{eq:analytic-square-root-update} equals $\norm{r_i}_2^2$, which is uniformly positive on $K_0$. Continuity and compactness give one product neighborhood and some $\bar\eta>0$ on which every denominator remains positive for $|\eta|<\bar\eta$. The positive square-root function is analytic on $(0,\infty)$, so the full state update is jointly real analytic there.

For $r_i\in\mathcal R_i$,
\[
\sum_a[F_i(r,\ell;g,\eta)]_a^2
=\frac{\sum_a r_i(a)^2e^{-\eta(2[\ell_g(r)]_i(a)-[\ell_i]_a)}}
       {\sum_b r_i(b)^2e^{-\eta(2[\ell_g(r)]_i(b)-[\ell_i]_b)}}
=1.
\]
The update also preserves nonnegativity, and $\ell_g(r)\in\mathcal V$; hence $G$ leaves $\mathcal S$ invariant. At $\eta=0$, the denominator equals $\norm{r_i}_2$, so $F_i(r,\ell;g,0)=r_i$ throughout the ambient neighborhood. Therefore $G(r,\ell;g,0)=(r,\ell_g(r))$. A second zero-step update leaves both components unchanged, proving the zero-step identity.
\end{proof}

\subsection{Analytic scaled finite differences}

For $s=(r,\ell)$, define the weighted-gap map by
\[
[H_{i,0}(s;g,\eta)]_{ab}
:=r_i(a)r_i(b)\bigl([\ell_i]_a-[\ell_i]_b\bigr),
\qquad a<b,
\]
and, for $\eta\ne0$, recursively define
\[
H_{i,k+1}(s;g,\eta)
:=\frac{H_{i,k}(G^2(s;g,\eta);g,\eta)
      -H_{i,k}(G(s;g,\eta);g,\eta)}{\eta}.
\]
The trajectory identity follows directly by induction. The case $k=0$ follows from $H_{i,0}(s^t;g,\eta)=h_i^t$. If it holds at order $k$, then
\begin{align*}
H_{i,k+1}(s^t;g,\eta)
&=\frac{H_{i,k}(s^{t+2};g,\eta)-H_{i,k}(s^{t+1};g,\eta)}{\eta}\\
&=\eta^{-(k+1)}
 \left(\Delta^kh_i^{t+k+2}-\Delta^kh_i^{t+k+1}\right)
 =\eta^{-(k+1)}\Delta^{k+1}h_i^{t+k+1}.
\end{align*}
Thus $H_{i,k}(s^t;g,\eta)=\eta^{-k}\Delta^k h_i^{t+k}$ for every $k\ge0$ and $t\ge1$.

The zero-step identity makes the numerator vanish at $\eta=0$. The next proof turns that cancellation into an analytic extension at every fixed order.

\printrepeatedstatement{lem-scaled-difference-analytic}{\scaledDifferenceAnalyticRestated}

\begin{proof}
We argue by induction and shrink the common neighborhood of $K_\star$ only when a composition requires it. The base map $H_{i,0}$ is polynomial. Suppose $H_{i,k}$ is analytic on such a neighborhood, and define
\[
N_{i,k}(s;g,\eta)
:=H_{i,k}(G^2(s;g,\eta);g,\eta)
 -H_{i,k}(G(s;g,\eta);g,\eta).
\]
By \Cref{lem:analytic-state-representation}, $N_{i,k}(s;g,0)=0$ on a neighborhood in $(s,g)$. After restricting to a product neighborhood on which $N_{i,k}$ is analytic, the fundamental theorem of calculus gives, coordinatewise,
\[
N_{i,k}(s;g,\eta)
=\eta\int_0^1\partial_\eta N_{i,k}(s;g,\tau\eta)\,\dd\tau.
\]
The integral is jointly analytic on a possibly smaller neighborhood and equals $N_{i,k}/\eta=H_{i,k+1}$ for $\eta\ne0$. It therefore extends $H_{i,k+1}$ analytically through $\eta=0$.
\end{proof}

\subsection{Frisch's Noetherianity Theorem}
\label{app:frisch-noetherian}

Frisch's theorem lets us select finitely many analytic generators near a compact semianalytic set. We introduce the definitions and three classical analytic lemmas needed for its proof. We state these lemmas in terms of scalar functions and analytic coefficients; their analytic foundations are supplied by the cited results. We reproduce the proof of Frisch's theorem using these lemmas.

\begin{definition}[Real-analytic functions]
\label{def:real-analytic}
\itshape
Let $U\subset\R^q$ be open. A function $f:U\to\R$ is \textit{real-analytic} if, for every $x\in U$, $f(x')$ equals a convergent power series in $x'-x$ for $x'$ in a neighborhood of $x$.
\end{definition}

\begin{definition}[Analytic germs]
\label{def:analytic-germs}
\itshape
Two real-analytic functions, each defined on an open neighborhood of a compact set $K\subset\R^q$, represent the same \textit{germ along $K$} if they agree on an open neighborhood of $K$ contained in both domains. A germ along $K$ is an equivalence class under this relation. The set of these germs is denoted by $\mathcal O(K)$.
\end{definition}

Intuitively, $\mathcal O(K)$ describes real-analytic functions near $K$, identifying two functions whenever they agree on an open neighborhood of $K$. In particular, $\mathcal O(K)$ forms a commutative ring: the sum and product of germs represented by $f$ and $g$ are represented by $f+g$ and $fg$ on the intersection of their domains. When $K=\{x\}$, write $\mathcal O_x:=\mathcal O(\{x\})$ and denote the germ of $f$ at $x$ by $f_x$. Note that agreement only at the points of $K$, rather than on a neighborhood of $K$, does not identify germs.

\begin{definition}[Ideals]
\label{def:analytic-ideals}
\itshape
A subset $I$ of a commutative ring $R$ is an \textit{ideal} if it contains zero and is closed under addition and multiplication by arbitrary elements of $R$.
\end{definition}

\begin{definition}[Finitely generated ideals]
\label{def:finitely-generated-ideal}
\itshape
An ideal $I\subseteq R$ is \textit{finitely generated} if there exist $f_1,\ldots,f_r\in I$ such that
\[
I=(f_1,\ldots,f_r)
:=\left\{\sum_{j=1}^r a_jf_j:a_1,\ldots,a_r\in R\right\}.
\]
\end{definition}

We write $f_x\in(f_{1,x},\ldots,f_{r,x})$ in $\mathcal O_x$ to denote that $f=\sum_j a_jf_j$ on some neighborhood of $x$, with analytic coefficients $a_j$ there.

\begin{definition}[Noetherian rings]
\label{def:noetherian-ring}
\itshape
A commutative ring $R$ is \textit{Noetherian} if every ideal of $R$ is finitely generated.
\end{definition}

\begin{definition}[Semianalytic sets]
\label{def:semianalytic}
\itshape
A set $S\subset\R^q$ is \textit{semianalytic} if, for every $x\in\R^q$, there exists an open neighborhood $U\subset\R^q$ of $x$ such that $S\cap U$ is a finite union of sets, each of the form
\[
\{x'\in U:f_1(x')=\cdots=f_r(x')=0,\quad g_1(x')>0,\ldots,g_s(x')>0\},
\]
where $r,s\ge0$ are integers and the functions $f_i,g_j:U\to\R$ are real analytic; either list of conditions may be empty. The neighborhood $U$ and the finite families of defining functions may depend on $x$.
\end{definition}

In particular, a set defined by finitely many polynomial equalities and weak inequalities is semianalytic, since a weak inequality is the union of a strict inequality and an equality.

The following three lemmas, \Cref{lem:analytic-local-generation}, \Cref{lem:analytic-membership-neighborhood}, and \Cref{lem:analytic-coefficient-gluing}, establish when $f$ can be written as $\sum_j a_jf_j$ with analytic coefficients $a_j$. They give finite generators at a point, a common neighborhood for local representations, and analytic coefficients valid near all of $K$.

We first seek finitely many generators at one point, which would allow us to write every germ in an ideal of $\mathcal O_x$ as their linear combination.

\begin{lemma}[Finite generation at a point {\citep[Theorem~1.15]{GreuelLossenShustin2007}}]
\label{lem:analytic-local-generation}
For every $x\in\R^q$, the ring $\mathcal O_x$ of real-analytic germs at $x$ is Noetherian.
\end{lemma}

The lemma implies that every germ in an ideal of $\mathcal O_x$ can be expressed using the same finite number of generators. Hence, in $f=\sum_{j=1}^r a_jf_j$, the coefficients $a_j$ may vary with $f$, but the generators $f_1,\ldots,f_r$ are shared across all real-analytic functions $f$ whose germs belong to that ideal.

The following lemma allows us to extend \Cref{lem:analytic-local-generation}'s finite generation at a single point to its neighborhood within $K$.

\begin{lemma}[Finite generation on a neighborhood {\citep[Chapter~V, \S3, pp.~186--187]{BanicaStanasila1976}}]
\label{lem:analytic-membership-neighborhood}
Let $K\subset\R^q$ be compact and semianalytic, let $x\in K$, and fix $f_1,\ldots,f_r\in\mathcal O(K)$. There is a relative open neighborhood $W$ of $x$ in $K$ such that, for every $f\in\mathcal O(K)$,
\[
f_x\in(f_{1,x},\ldots,f_{r,x})
\quad\Longrightarrow\quad
f_{x'}\in(f_{1,x'},\ldots,f_{r,x'})
\quad\text{for every }x'\in W.
\]
The neighborhood $W$ may depend on $K,x,f_1,\ldots,f_r$, but not on $f$.
\end{lemma}

Intuitively, \Cref{lem:analytic-local-generation} selects finitely many generators from an ideal at $x$, and \Cref{lem:analytic-membership-neighborhood} gives a single neighborhood where those generators work locally for every element of the ideal. Compactness then allows us to collect these generators into a single finite family that works at every point of $K$. The following lemma turns these local representations into a representation with analytic coefficients valid near all of $K$.

\begin{lemma}[Analytic coefficients near a compact set {\citep[Theorem~7.2.1(ii)]{Hormander1990}}]
\label{lem:analytic-coefficient-gluing}
Let $K\subset\R^q$ be compact and let $f,f_1,\ldots,f_r\in\mathcal O(K)$. If $f_x\in(f_{1,x},\ldots,f_{r,x})$ for every $x\in K$, then there are $a_1,\ldots,a_r\in\mathcal O(K)$ such that
\[
f=\sum_{j=1}^r a_jf_j
\quad\text{in }\mathcal O(K).
\]
\end{lemma}

Together with compactness, these three lemmas yield Frisch's finite-generation theorem.

\begin{theorem}[Frisch's Noetherianity theorem, real-analytic form {\citep[Theorem~I.9]{Frisch1967}}]
\label{thm:frisch-noetherian}
Let $K\subset\R^q$ be compact and semianalytic. Then $\mathcal O(K)$ is Noetherian: every ideal $I\subseteq\mathcal O(K)$ has finitely many generators $f_1,\ldots,f_r\in I$, so that
\[
I=\left\{\sum_{j=1}^r a_jf_j:a_1,\ldots,a_r\in\mathcal O(K)\right\}.
\]
\end{theorem}

\begin{proof}
Fix an ideal $I\subseteq\mathcal O(K)$; the case $K=\varnothing$ is immediate. We first choose generators that work locally throughout $K$, and then obtain analytic coefficients near the whole set.

\emph{Step 1: finite generators near each point.}
For $x\in K$, let $I_x\subseteq\mathcal O_x$ be the ideal generated by $\{f_x:f\in I\}$. By \Cref{lem:analytic-local-generation}, $I_x$ has finitely many generators. Each generator is a finite combination of germs coming from $I$. Collecting the elements of $I$ in these combinations gives $f_{x,1},\ldots,f_{x,r_x}\in I$ whose germs generate $I_x$. If $I_x=\{0\}$, take the single generator $0$.

Apply \Cref{lem:analytic-membership-neighborhood} to this fixed finite family. It gives a relative open neighborhood $W_x$ of $x$ in $K$ such that
\[
f_{x'}\in(f_{x,1,x'},\ldots,f_{x,r_x,x'})
\qquad\text{for every }f\in I\text{ and every }x'\in W_x.
\]
The same $W_x$ works for every $f\in I$. Thus the selected elements generate $I_{x'}$ at every $x'\in W_x$.

\emph{Step 2: one finite family for the compact set.}
The sets $W_x$ cover $K$. Choose a finite subcover and collect the corresponding generators into $f_1,\ldots,f_r\in I$. For every $x\in K$, Step~1 gives
\[
I_x=(f_{1,x},\ldots,f_{r,x})\subseteq\mathcal O_x.
\]
The reverse inclusion needed for this equality holds because every selected $f_j$ belongs to $I$.

\emph{Step 3: coefficients on a neighborhood of $K$.}
Fix any $f\in I$. By Step~2, $f_x\in(f_{1,x},\ldots,f_{r,x})$ for every $x\in K$. \Cref{lem:analytic-coefficient-gluing} gives $a_1,\ldots,a_r\in\mathcal O(K)$ with $f=\sum_j a_jf_j$. Hence $I\subseteq(f_1,\ldots,f_r)$, and the reverse inclusion follows from the ideal property. This proves that $I$ is finitely generated.

\end{proof}

\subsection{Exact finite-order relation}
\label{app:finite-generation}

Frisch's theorem gives finite generation for scalar analytic functions. We first extend Frisch's theorem to vector-valued maps and prove \Cref{thm:finite-dependence}. Then, we apply \Cref{thm:finite-dependence} to the scaled finite differences and obtain \Cref{lem:finite-gap-relation}.

\printrepeatedstatement{thm-finite-dependence}{\finiteDependenceRestated}

\begin{proof}
Put $R:=\mathcal O(K)$ and $d_V:=\dim V$. By \Cref{thm:frisch-noetherian}, every ideal of $R$ is finitely generated. We first verify the corresponding fact for a submodule $L\subseteq R^{d_V}$, a set of vectors containing zero and closed under addition and multiplication by elements of $R$. Induct on $d_V$. The case $d_V=0$ is immediate, and the case $d_V=1$ is the ideal property. For $d_V\ge2$, projection onto the last coordinate sends $L$ to an ideal of $R$. Choose finitely many generators of that ideal and lift them to vectors $v_1,\ldots,v_r\in L$. The kernel is a submodule of $R^{d_V-1}$, so it has finitely many generators by induction. For any $v\in L$, subtracting a suitable $R$-linear combination of $v_1,\ldots,v_r$ makes its last coordinate zero. The remainder is generated by the kernel generators. These two families therefore generate $L$.

Choose a basis of $V$, and write $[\Phi_k]\in R^{d_V}$ for the germ of $\Phi_k$ along $K$. Consider the submodule
\[
M:=\sum_{k\ge0}R[\Phi_k]\subseteq R^{d_V},
\]
where every sum defining an element of $M$ is finite. By the preceding argument, $M$ has finitely many generators. Each is a finite combination of the $[\Phi_k]$, so there is an $N\ge0$ with $M=\sum_{k=0}^N R[\Phi_k]$. If $M=\{0\}$, take $N=0$. Setting $m:=N+1$ gives
\[
[\Phi_m]=\sum_{k=0}^{m-1}[c_k][\Phi_k]
\qquad\text{for some }[c_0],\ldots,[c_{m-1}]\in R.
\]
Only finitely many maps and coefficient germs enter this equality. Choose their representatives on the intersection of their open neighborhoods of $K$. Equality as germs means that, after a further restriction to an open neighborhood $U\supseteq K$, the equality holds pointwise throughout $U$. This proves \Cref{eq:finite-dependence}. The argument does not require one common neighborhood for the entire infinite sequence.
\end{proof}

The loss table is part of the compact set $K_\star$. This lets us choose the finite relation once for the entire dimension class, before fixing a game or trajectory.

\printrepeatedstatement{lem-finite-gap-relation}{\finiteGapRelationRestated}

\begin{proof}
Fix the player $i$. Finitely many polynomial equalities and weak inequalities define $K_\star=K_0\times\{0\}$, where $K_0=\mathcal R\times\mathcal V\times\mathcal G$. Thus $K_\star$ is compact and semianalytic. By \Cref{lem:scaled-difference-analytic}, each $H_{i,k}$ is analytic on a neighborhood of $K_\star$. Apply \Cref{thm:finite-dependence} to the sequence $(H_{i,k})_{k\ge0}$ to obtain an integer $m\ge1$ and analytic scalar coefficients $c_0,\ldots,c_{m-1}$ such that
\[
H_{i,m}=\sum_{k=0}^{m-1}c_kH_{i,k}
\]
on a common open neighborhood $U$ of $K_\star$. Intersect $U$ with the domains of $G,H_{i,0},\ldots,H_{i,m}$ if needed; there are only finitely many of these domains.

Because $K_0$ is compact and $U$ is open, some $\eta_i>0$ satisfies
\begin{equation}
K_0\times[-\eta_i,\eta_i]\subseteq U.
\label{eq:compact-product-set}
\end{equation}
Continuity on this compact set gives finite constants $A,M$ such that, for every $0\le k<m$, $|c_k|\le A$, and, for every $0\le k\le m$, $\norm{H_{i,k}}_2\le M$. The compact set and analytic maps were fixed by $(n,d)$ and $i$ before any particular loss table was chosen. Hence $m,\eta_i,A,M$ are uniform over all games of these dimensions and do not depend on the horizon.

Fix any such game, $0<\eta\le\eta_i$, and $t\ge1$. Invariance of $\mathcal S$ gives $(s^t,g,\eta)\in K_0\times[-\eta_i,\eta_i]$. Evaluating the relation there and using the trajectory identity from \Cref{eq:scaled-difference-functions} yields
\[
\eta^{-m}\Delta^m h_i^{t+m}
=H_{i,m}(s^t;g,\eta)
=\sum_{k=0}^{m-1}c_k(s^t;g,\eta)
 \eta^{-k}\Delta^kh_i^{t+k}.
\]
Set $c_{k,t}:=c_k(s^t;g,\eta)$ and multiply by $\eta^m$ to obtain \Cref{eq:exact-relation-main}, with $|c_{k,t}|\le A$. For $0\le k\le m$, the same trajectory identity and the bound on $H_{i,k}$ give
\[
\norm{\Delta^kh_i^{t+k}}_2
=\eta^k\norm{H_{i,k}(s^t;g,\eta)}_2
\le M\eta^k,
\]
which is \Cref{eq:scaled-difference-size-main}.
\end{proof}

\section{Finite-support cutoff construction}
\label{app:technical}

To localize the trajectory, we use a gradual cutoff rather than zero-padding. This preserves the $\eta^m$ scale of the recurrence and controls the cost of returning to the original segment.

\printrepeatedstatement{lem-exact-boundary-extension}{\finiteSupportCutoffRestated}

\begin{proof}
For short horizons, we pay for the whole segment directly. Otherwise, we keep a central interval unchanged and taper over $O(1/\eta)$ indices at both ends; a buffer before the first taper avoids the exceptional initial first-order difference. Choose a constant $C_0$ depending only on $m,A,M$. If $T\le C_0/\eta$, take $y^t=0$ for every $t$. The $k=0$ case of \Cref{eq:exact-boundary-assumption-main} gives $\norm{u^t}\le M$. Hence $\norm{\Delta u^1}\le2M$, while the $k=1$ case gives $\norm{\Delta u^t}\le M\eta$ for every $t\ge2$. Therefore
\[
\sum_{t=1}^T\norm{\Delta u^t}^2
\le4M^2+C_0M^2\eta
\le C.
\]
The two norm comparisons follow, and \Cref{eq:exact-norm-recurrence-main} is immediate because $y=0$.

Suppose now that $T>C_0/\eta$. Put $L:=\lceil1/\eta\rceil$. Choose a fixed smooth function $\psi:\R\to[0,1]$ that is zero on $(-\infty,0]$, one on $[1,\infty)$, and has bounded derivatives. Define
\begin{equation}
\chi^t
:=\psi\!\left(\frac{t-1-3m}{L}\right)
  \psi\!\left(\frac{T+1-3m-t}{L}\right),
\label{eq:exact-boundary-weights-app}
\end{equation}
set $y^t:=\chi^tu^t$ for $1\le t\le T+1$, and set $y^t:=0$ otherwise. The construction keeps $y^t=u^t$ on the central interval
\[
1+3m+L\le t\le T+1-3m-L
\]
and is zero before and after the two intervals on which $\chi^t$ changes. Let $B_T$ be the set of starting indices $t$ for which $\{t,\ldots,t+2m\}$ intersects either of these two intervals. Increasing $C_0$ if necessary makes the central interval nonempty and gives
\[
|B_T|\le C(L+m)\le C/\eta.
\]

Repeated use of
\[
\Delta\psi(t/L)=\int_{t/L}^{(t+1)/L}\psi'(s)\,\dd s
\]
and the finite-difference product rule gives $\sup_t|\Delta^\ell\chi^t|\le C_{m,\ell}\eta^\ell$ for $0\le\ell\le m$. Moreover,
\begin{equation}
\Delta^j(\chi u)^t
=\sum_{\ell=0}^j\binom j\ell
 (\Delta^\ell\chi)^t(\Delta^{j-\ell}u)^{t+\ell}.
\label{eq:boundary-product}
\end{equation}
The buffer of length $3m$ before the first transition ensures that every finite difference of $u$ of positive order in \Cref{eq:boundary-product} begins at an index covered by \Cref{eq:exact-boundary-assumption-main}. Thus, whenever $t\in B_T$,
\begin{equation}
\norm{\Delta^j y^{t+r}}
\le C\eta^j,
\qquad 0\le j\le m,
\quad 0\le r\le2m-j.
\label{eq:exact-boundary-local-app}
\end{equation}

We retain the original recurrence wherever the cutoff leaves the trajectory unchanged and collect its failure into a boundary residual. For each $t$, define
\[
\widetilde c_{k,t}:=
\begin{cases}
 c_{k,t},&t\ge1\text{ and }y^{t+r}=u^{t+r}\text{ for every }0\le r\le2m,\\
 0,&\text{otherwise},
\end{cases}
\]
and define $e=(e^t)_{t\in\mathbb Z}$ by
\begin{equation}
\Delta^my^{t+m}
=\sum_{k=0}^{m-1}\widetilde c_{k,t}\eta^{m-k}\Delta^ky^{t+k}+e^t.
\label{eq:constant-perturbed-recurrence-app}
\end{equation}
If $t\notin B_T$, the values $y^t,\ldots,y^{t+2m}$ either agree with the corresponding values of $u$ or are all zero. In the first case \Cref{eq:exact-boundary-recurrence-main} gives $e^t=0$; in the second case this is immediate. Hence $e^t=0$ for every $t\notin B_T$. For $t\in B_T$, \Cref{eq:exact-boundary-local-app} bounds every term in \Cref{eq:constant-perturbed-recurrence-app} by $C\eta^m$. Therefore
\begin{equation}
\norm e_{\ell_2}
\le C\eta^m\sqrt{|B_T|}
\le C\eta^{m-1/2}.
\label{eq:constant-boundary-error-app}
\end{equation}
Taking $\ell_2$ norms in \Cref{eq:constant-perturbed-recurrence-app}, using $|\widetilde c_{k,t}|\le A$, shift invariance, and \Cref{eq:constant-boundary-error-app}, yields \Cref{eq:exact-norm-recurrence-main}.

Finally, $0\le\chi^t\le1$ gives \Cref{eq:cutoff-sequence-norm}. For $1+3m+L\le t\le T-3m-L$, both $t$ and $t+1$ lie in the central interval, so $\Delta y^t=\Delta u^t$. Among $2\le t\le T$, at most $C(L+m)\le C/\eta$ indices lie outside this range, including the last index of the central interval. For these indices, \Cref{eq:exact-boundary-assumption-main} gives $\norm{\Delta u^t}\le M\eta$, so their squared contribution is at most $C\eta$. The initial difference satisfies $\norm{\Delta u^1}\le2M$. These bounds prove \Cref{eq:exact-boundary-return-main}.
\end{proof}

\end{document}